\RequirePackage{fix-cm}
\documentclass[smallextended]{svjour3}       
\smartqed  
\usepackage{graphicx}
\usepackage{textcomp}
\usepackage{hyperref}
\usepackage[show]{ed}
\usepackage{amssymb}
\usepackage{amsmath}
\usepackage{comment}
\usepackage{topcapt}
\usepackage{algorithm}
\usepackage[noend]{algpseudocode}
\usepackage{color, colortbl}

\definecolor{Gray}{gray}{0.9}
\definecolor{Green}{rgb}{152,251,152}

\DeclareMathOperator*{\argmax}{arg\,max}
\DeclareMathOperator*{\argmin}{arg\,min}
\DeclareMathOperator{\vectorization}{Vec}
\renewcommand{\vec}[1]{\mathbf{#1}}

\begin{document}

\title{Exact and efficient top-$K$ inference for multi-target prediction by querying separable linear relational models}

\titlerunning{Efficient top-$K$ inference in multi-target prediction}        

\author{Michiel Stock         \and
        Krzysztof Dembczy\'nski \and Bernard De Baets \and Willem Waegeman
}


\institute{Michiel Stock, Bernard De Baets and Willem Waegeman \at
				KERMIT, Department of Mathematical Modelling, Statistics and Bioinformatics \\ 
Ghent University, 9000 Ghent, Belgium \\
\email{firstname.lastname@ugent.be}
\and
Krzysztof Dembczy{\'n}ski \at 
 Institute of Computing Science \\
Poznan University of Technology, Poznan 60-695, Poland \\
           \email{krzysztof.dembczynski@cs.put.poznan.pl}
}

\date{Received: date / Accepted: date}

\maketitle

\begin{abstract}
{
Many complex multi-target prediction problems that concern large target spaces are characterised by a need for efficient prediction strategies that avoid the computation of predictions for all targets explicitly. Examples of such problems emerge in several subfields of machine learning, such as collaborative filtering, multi-label classification, dyadic prediction and biological network inference. In this article we analyse efficient and exact algorithms for computing the top-$K$ predictions in the above problem settings, using a general class of models that we refer to as separable linear relational models. We show how to use those inference algorithms, which are modifications of well-known information retrieval methods, in a variety of machine learning settings. Furthermore, we study the possibility of scoring items incompletely, while still retaining an exact top-$K$ retrieval. Experimental results in several application domains reveal that the so-called threshold algorithm is very scalable, performing often many orders of magnitude more efficiently than the naive approach.
}

\end{abstract}

\section{Introduction and formal problem description}
\label{intro}
Subjected to both great expectations as well as great criticism, ``Big Data" and ``Big Data Analytics" are two notions that are without doubt highly popular these days. Not only does ``Big Data" generate a number of interesting research questions for classical machine learning settings such as classification and regression problems, it also leads to novel challenges in less traditional learning settings. For example, though multi-target prediction is not a recent field (see e.g.~\cite{Blockeel1998}), Big Data problems nowadays routinely deal with very large output spaces. Roughly speaking, multi-target prediction can be seen as a term that intends to unify methods that are developed in several subfields of machine learning. Multi-target prediction has as general research theme the principles behind learning problems where predictions for multiple targets need to be generated simultaneously\footnote{We refer to the ICML 2013 tutorial and the ECML-PKDD 2014 workshop on multi-target prediction for an overview.}. 
In this work we will use a general notation $s(x,y)$ to denote the predicted score of a multi-target prediction model, where $x$ will be referred to as an instance or query, and $y$ as a target, a label or an item. We give a couple of examples that are analyzed in more detail later in this article to make this point more clear: 
\begin{itemize}
\item In collaborative filtering methods for recommender systems, $s(x,y)$ will represent a score that indicates the degree to which query $x$ will be interested in item $y$. Using methods such as matrix factorization, users and items are here typically represented by a set of latent variables -- see e.g.\ \cite{Takacs2008,Volinsky2009}. 
\item In multi-label classification and multivariate regression (often denoted as multi-output regression), $s(x,y)$ will represent the prediction for instance $x$ and label $y$, where typically a feature representation is available for the user, whereas no additional information about $y$ is known in the basic setting -- see e.g. \cite{Agrawal2013,Dembczynski2012,Tsoumakas2007}.
\item In content-based filtering, dyadic prediction and network inference problems, feature vectors will be available for both instances $x$ and targets $y$. For example, in content-based filtering, one would recommend items to users based on user profiles and side information about items \cite{Basilico2004a,chu2009}. In protein-ligand interaction modelling for drug design, a bio-informatics application of dyadic prediction, one would model a compatibility score $s(x,y)$ based on feature descriptions of proteins and ligands \cite{Jacob2008,Wang2015}. 
\end{itemize}

Using a generic methodology, we analyze in this article efficient methods for finding the best scoring targets in the above application domains. Existing machine learning methods in those domains are often suffering from severe bottlenecks when predictions need to be generated and stored for a large number of couples $(x,y)$. As a result, instead of computing the scores $s(x,y)$ explicitly for all couples, one could argue that it suffices to retrieve the objects $y$ resulting in the highest scores for a given object ${x}$. For example, in computational drug design it suffices to retrieve the best-binding molecules for a given protein -- see e.g.~\cite{Jacob2008a}. Similarly, in marketing applications of recommender systems, one is often mainly interested in those products that are most relevant for a given user -- see e.g.~\cite{Su2009}. More formally, rephrased as an information retrieval problem, we are often interested in finding the set that contains the $K$ highest-scoring objects of this database with respect to a certain instance (or query)~$x$. 

Let us introduce some further notation to make this problem statement a bit more precise. Using a generic notation, we consider two types of objects $x\in\mathcal{X}$ and $y\in\mathcal{Y}$. For simplicity, we assume that $\mathcal{Y}$ is finite with cardinality~$M$. We will analyse bilinear models that calculate for each couple $(x,y)$ the following score as prediction:
\begin{equation}
s(x,y) = \vec{u}(x)^\intercal  \vec{t}(y) = \sum_{r=1}^R u_r(x) \, t_r(y) \label{s_in_tu} \,.
\end{equation} 
As the representations need to be separated for $x$ and $y$, we call the above class of models separable linear relational (SEP-LR) models. Each of the two objects is represented by an $R$-dimensional model vector:
\begin{eqnarray*}
\vec{u}(x) &=& (u_1(x), u_2(x), \ldots, u_R(x))^\intercal \,, \\
\vec{t}(y) &=& (t_1(y), t_2(y), \ldots, t_R(y))^\intercal \,. 
\end{eqnarray*}
The general applicability of SEP-LR models to the application domains that are mentioned above will be further discussed in Section~\ref{SLRapplications}. In matrix factorization methods, $R$ will correspond to the rank of the low-rank decomposition. In multi-label classification problems, it will be the number of features. 

 As we are interested in computing only a subset of the predictions for couples $(x,y)$, we define the set $S_x^K$, containing the $K$ most relevant objects $y$ for the query $x$. The problem that we intend to solve can then be formally written as: 
 \begin{eqnarray}
 \label{eq:prob_statement}
S_x^K &=& \argmax_{\mathcal{S} \in [\mathcal{Y}]^K} \min_{y \in \mathcal{S}} s(x,y) \,,
 \end{eqnarray}
with $[\mathcal{Y}]^K$ the set of all $K$-element subsets of $\mathcal{Y}$. Thus, for two objects $y\in S^K_{x}$ and $y'\notin S^K_{x}$, it holds that $s(x, y) \geq s(x, y')$. The set $S_x^K$ is not necessarily unique, because ties may occur. Let us remark that problem statement (\ref{eq:prob_statement}) is related to nearest neighbor search, but finding the point with maximum inner product or maximum cosine similarity is not identical to finding the point that is closest w.r.t.\ Euclidean distance:
\begin{eqnarray*} 
\argmin_{y \in \mathcal{Y}} || \vec{u}(x) -  \vec{t}(y) ||^2 &=& \argmin_{y \in \mathcal{Y}} ||\vec{t}(y)||^2 - 2 \vec{u}(x)^\intercal  \vec{t}(y) +||\vec{u}(x)||^2 \\
&=& \argmin_{y \in \mathcal{Y}} ||\vec{t}(y)||^2 - 2 \vec{u}(x)^\intercal  \vec{t}(y)\,.
\end{eqnarray*} 
The term $||\vec{u}(x)||^2$ can be dropped as it remains constant as soon as $x$ is fixed, but $||\vec{t}(y)||^2$ cannot be omitted. {When the norm of $\vec{t}(y)$ is the same for each $y$, this expression is equivalent with (\ref{s_in_tu}). In many applications, the norm of the items has a clear meaning and cannot be omitted. Also, in contrast to distances, we study more general problem settings in which $x$ and $y$ belong to a different domain and a distance between them is less natural to define.} Existing methods for speeding up nearest neighbor search, such as \cite{Elkan2003}, are hence not directly applicable to maximum inner product search.  

One approach to compute the maximum inner product would be to partition the target space using efficient data structures such as $k$-d trees~\cite{Bentley1975}, ball trees~\cite{Omohundro1989}, cover trees~\cite{Beygelzimer2006} or branch-and-bound search techniques \cite{Koenigstein2012}. Methods of that kind are able to find the top-scoring predictions in an efficient way for low-dimensional Euclidean embeddings (in our case, when $R$ is small), but they bring no improvement compared to a naive linear search when the dimensionality is larger than twenty. Space partitioning methods are hence inapplicable to the problem settings that are the focus of this paper. 

When time efficiency is more important than predictive power, one could opt for employing specialized approximate algorithms. Locality-sensitive hashing techniques, which are popular for the related problem of nearest neighbor search, would be a good choice. Recently, a locality-sensing hashing method for maximum inner product search has been developed \cite{Shrivastava2014,Shrivastava2015}. However, in addition to delivering approximate predictions, methods of that kind are usually restricted to finding the top-1 set. Another approximate method would be to cluster the queries $x$ in several groups, for which rankings of targets $y$ can be precomputed by means of predictive indices and related data structures \cite{Agarwal2012,Goel2009}.   

In this article we are interested in computing the set $S_x^K$ as efficiently as possible in an exact manner. To this end, we depart from a trained model $s(x,y)$ and we make a number of additional assumptions that are all quite natural for the applications we have in mind. First, we assume that the top-$K$ predictions should be returned, where $K\geq1$. Second, we also assume that $R$ can be bigger than twenty, so that space partitioning methods become inapplicable. Third, we assume that at prediction time queries $x$ need to be processed one-by-one. If queries are arriving in large batches, further speed-ups could be obtained by using specialized libraries for matrix multiplication. 

We analyse exact algorithms for inferring the top-$K$ set by adopting existing methods from the database and information retrieval literature. Indeed, a strongly related problem as (\ref{eq:prob_statement}) is often observed in information retrieval~\cite{Ilyas2008}. When queries and documents have sparse representations and relevance is defined by means of cosine similarity, one can reduce the computational complexity of retrieval by using data structures such as inverted indices \cite{Zobel2006}. We will exploit a similar idea. Even though most of the specialized techniques in information retrieval put a strong emphasis on sparsity, which makes them inapplicable to multi-target prediction problems, we will show that certain techniques can be used to compute the top-$K$ set more efficiently compared to linear search. In particular, we are analyzing in this paper Fagin's algorithm and extensions thereof \cite{Fagin1999}. Similar to inverted indices, those algorithms will score for a given query the items of several lists until one is guaranteed to have found the top-$K$ scoring items. More details are provided below.

The remainder of this work is structured as follows. We present and discuss in Section~\ref{queryAlgorithms} three exact algorithms for obtaining the $K$ highest scores for a given query. In Section~\ref{SLRapplications} we will give an overview of machine learning problems that can be cast in our general framework, followed by Section~\ref{experiments}, where we experimentally study the algorithms using datasets for several application domains. Finally, in Section~\ref{Faginconclfutwork}, we conclude with some practical guidelines and some directions for future research.

\section{Exact algorithms for top-$K$ inference}\label{queryAlgorithms}

The most straightforward way to solve problem (\ref{eq:prob_statement}) is by simply calculating the score $s(x,y)$ for each target $y$ and retaining the $K$ targets with the highest scores. 
We will call this algorithm the naive algorithm, as we have to calculate all $M$ scores to gather the top-$K$ highest scoring targets. Computing the score for one object has a time and space complexity of $\mathcal{O}(R)$, as this amounts to computing a weighted sum of $R$ terms. Apart from some set operations that can be done with a time complexity of $\mathcal{O}(1)$, we also occasionally have to update the current top-$K$ set when a new target is scored higher than the worst target in the current list. Using efficient data structures such as heaps, this can be done with a time complexity of $\mathcal{O}(\log K)$. Hence the time complexity for the naive algorithm is $\mathcal{O}((R+\log K)M)$.

We show that the problem can be solved more efficiently using exact methods that are well known in database research and information retrieval, namely Fagin's algorithm \cite{Fagin1999} and the so-called threshold algorithm \cite{Fagin2003}. Both algorithms in essence find in an efficient way the maximum of an aggregation operator $Q(z_1,...,z_R)$. More specifically, if $z_1,...,z_R$ are the $R$ grades of an object, then $Q(z_1,...,z_R)$ represents the overall grade of that object. Fagin's algorithm and the threshold algorithm both find in an exact but efficient manner the $K$ objects in a database with highest $Q(z_1,...,z_R)$. The algorithms assume that all variables $z_1,...,z_R$ belong to the interval $[0,1]$ and $Q$ has to be an increasing aggregation operator, i.e., $Q(z_1,...,z_R) \leq Q(z_1',...,z_R')$ if $z_r \leq z_r'$ for every $r$. 

Given a query $x$, problem statement (\ref{eq:prob_statement}) can be easily transformed into the original problem setting of Fagin by defining $z_r = u_r(x) t_r(y)$. The resulting values $z_r$ do not necessarily belong to the interval $[0,1]$, but they can be transformed accordingly for a fixed query $x$ and set of targets $\mathcal{Y}$. In what follows we therefore assume that Fagin's algorithm and the threshold algorithm are applicable, and we further adopt a machine learning notation in explaining the two algorithms.  

Key to both algorithms is a set of $R$ sorted lists $L_1,...,L_R$ that contain pointers to all the targets, ordered according to each of the $R$ model descriptors $t_i(y)$. The pseudo-code of Fagin's algorithm is given in Algorithm~\ref{FaginAlgorithm}. It first scans the targets that are at the top of the selected lists until $K$ targets are found that occur somewhere close to the top of each list (the random access phase). Those targets are stored in the set \texttt{targetsToCheck}. Subsequently, the score for all observed targets is calculated in a second phase, and the set $S_x^K$ is constructed as the $K$ targets with highest scores in this set (the sorted access phase).

It should be intuitively clear that the set $S_x^K$ will be computed in a correct manner. The stopping criterion implies that for at least $K$ targets the score $s(x,y)$ has been computed. The best targets in \texttt{targetsToCheck} will be returned as the set $S_x^K$, while the monotonicity of the scoring function guarantees that the scores for these $K$ targets are at least as high as for any target not in the considered top of the sorted lists. If ${u_r}(x)$ is negative, the corresponding list $L_r$ is reverted in the first part of the algorithm. This is equivalent to  transferring the sign of ${u}_r(x)$ to the corresponding features of $t_r(y)$, i.e.\ working with $|{u}_r(x)|$ and $-t_r(y)$. Without loss of generality, we can thus assume that the score $s(x,y)$ is increasing w.r.t.\ all $t_r(y)$.

\begin{algorithm}[t]
\caption{Fagin's Algorithm}\label{FaginAlgorithm}
\begin{algorithmic}[1]

\Statex \textbf{Input}: $\mathcal{Y}$, $x$, $K$, $L_1,\ldots,L_R$
\Statex \textbf{Output}: $S^K_x$
\Statex 

\State $S^K_x \leftarrow \emptyset$

\State bookkeeping[1..$M$] $\leftarrow0$
\State targetsToCheck $\leftarrow0$
\State numberOfTargetsInAllLists $\leftarrow0$

\While{numberOfTargetsInAllLists $< K$} \label{faginrandomstart}
	\For{$r\leftarrow1$ \textbf{to} $R$} 
		\State $y \leftarrow$ get next item from $L_r$
		\State targetsToCheck $\leftarrow$ targetsToCheck $\cup\ \{y\} $
		\State bookkeeping[$y$] $\leftarrow$ bookkeeping[$y$]+1
		\If{bookkeeping[$r$] $=R$}
			\State numberOfTargetsInAllLists $\leftarrow$ numberOfTargetsInAllLists $+ 1$
		\EndIf
	\EndFor
\EndWhile \label{faginrandomend}

\For{$y$ \textbf{in} targetsToCheck}
	\State score $\leftarrow s(x,y)$ \Comment{$\mathcal{O}(R)$}
	\If{lowest score in $S^K_x<$ score}
		\State update $S^K_x$ with scored item $y$  \Comment{$\mathcal{O}(\log K)$}
	\EndIf
\EndFor

\end{algorithmic}
\end{algorithm}

\begin{algorithm}[t]
\caption{Threshold Algorithm}\label{ThresholdAlgorithm}
\begin{algorithmic}[1]
\Statex \textbf{Input}: $\mathcal{Y}$, $x$, $K$, $L_1,\ldots,L_R$
\Statex \textbf{Output}: $S^K_x$
\Statex 

\State $S^K_x \leftarrow \emptyset$

\State  calculated $\leftarrow \emptyset$ 
\State  lowerBound $\leftarrow-\infty$
\State  upperBound $\leftarrow+\infty$

\While{lowerBound $<$ upperBound}
	\State  upperBound $\leftarrow0$
	\For{$r\leftarrow1$ \textbf{to} $R$}\label{sparselists}
		\State $y \leftarrow$ get next item from $L_r$
		\If{$y \not\in$ calculated}
			\State score $\leftarrow s(x,y)$ \Comment{$\mathcal{O}(R)$}
			\State calculated $\leftarrow$ calculated $\cup\ \{y\}$
	\If{lowerBound $<$ score}
		\State update $S^K_x$ with the new scored item  \Comment{$\mathcal{O}(\log K)$}
		\If{$|S^K_x| = K$}
			\State lowerBound $\leftarrow$ lowest score of a target in $S_x^K$ \Comment{$\mathcal{O}(1)$}
		\EndIf
	\EndIf
		\EndIf
		\State upperBound $\leftarrow$ upperBound + $u_r(x) \, t_r(y)$ \Comment{$\mathcal{O}(1)$}

	\EndFor
\EndWhile
\end{algorithmic}
\end{algorithm}

The pseudo-code of the threshold algorithm is given in Algorithm~\ref{ThresholdAlgorithm}. In contrast to Fagin's algorithm, this approach uses information of the query $x$ to put emphasis on the dimensions in $\vec{t}(y)$ that are relevant for $\vec{u}(x)$. The algorithm is therefore not divided in a random and a sorted access phase. In iteration $d$, it scores the targets observed at depth $d$ in the lists. It keeps popping elements from lists to obtain promising targets until a stopping criterion is reached, i.e., when the lowest score in the current top-$K$ set, the lower bound, exceeds an upper bound on the values of the scores of targets that have not been investigated yet. This is summarized in the following theorem. 

\begin{theorem}\label{ubtheorem}
Upon termination of the threshold algorithm, it has found the set $S_x^K$. 
\end{theorem}
\begin{proof}
In iteration $d$, let $y_{L_r(d)}$ be the target at position $d$ of list $L_r$, then the upper bound is given by:
\begin{equation}
\mathrm{upperBound}(d) = \sum_{r=1}^R u_r(x) \, t_r(y_{L_r(d)})\,. \label{upperboundTA}
\end{equation}
For a given query $x$, if an item $y$ has not occurred in any of the lists at depth $d$, its score $s(x,y)$ will not exceed $\mathrm{upperBound}(d)$. As the score can be calculated as an increasing function of the components of $t_r(y)$, it follows that $\mathrm{upperBound}(d)$ is greater than or equal to the score of any target not yet encountered. As a result, this is an upper bound on the scores of such targets, so it should be clear that the algorithm computes the set $S_x^K$ in a correct manner. \qed
\end{proof}

\begin{table}[htbp]
   \centering
   \topcaption{Top: Toy dataset for a given query $x$ with $R=4$, $K=1$ and $\mathbf{u}(x)=(0.1,  2.5,  1 ,  0.5)^\intercal$. The best item is no.\ 6 and it is indicated in boldface. Left: Fagin's algorithm applied to the toy dataset for finding the top-1 set. Item no.\ 5 has appeared in each list at depth five, so Fagin's algorithm will stop after five iterations. We score all encountered items: $\{2, 3, 4, 5, 6, 7, 8, 9, 10\}$. The top is calculated by scoring nine items instead of ten and item 6 is returned as best scoring item. Right: Applying the threshold algorithm to the toy example. In the second step the lower bound already exceeds the upper bound and the algorithm terminates. The top-scoring target is found while only five of the ten targets are scored.}\label{FaginExample}
	   \begin{tabular}{| c | cccc | c |} 
   \hline
item & $t_1(y)$ & $t_2(y)$ & $t_3(y)$ & $t_4(y)$ & $s(x, y)$\\
\hline
1 &-0.5 & -1.4 & -0.8 & -1.0  & -4.85\\
2 & 0.9 & -1.9 & -0.3 &  0.5 & -4.71\\
3 &-0.8 & -0.4 & -0.1 &  0.9 & -0.73\\
4 &-0.7 & -1.7 &  0.2 & -2.5 & -5.37\\
5 & 0.8 &  0.2 &  0.0  &  0.7 &  0.93\\
6 & 1.0  &  1.6 &  0.9 & -0.6 &  \textbf{4.7} \\
7 & 0.1 &  0.4 & -0.6 & -2.0  & -0.59\\
8 & -2.4 &  0.6 &  0.4 & -0.4 &  1.46\\
9 & -1.6 &  0.2 &  1.0  &  0.3 &  1.49\\
10 & 0.0  &  1.0  & -0.6 &  1.4 &  2.6 \\ 
\hline
   \end{tabular} \\
  \vspace{0.3cm} 
	\begin{tabular}{|c| cccc |} 
   \hline
step & $L_1$ & $L_2$ & $L_3$ & $L_4$\\
\hline
\rowcolor{Gray}
      1 & 6&  6&  9& 10\\
\rowcolor{Gray}    
   2& 2& 10&  6&  3\\
   \rowcolor{Gray}
      3&  \textbf{5}&  8&  8&  \textbf{5}\\
       \rowcolor{Gray}
       4&7&  7&  4&  2\\
        \rowcolor{Gray}
       5 &10&  \textbf{5}&  \textbf{5}&  9\\
       - & 1&  9&  3&  8\\
       -& 4&  3&  2&  6\\
       -& 3&  1&  7&  1\\
       -& 9&  4& 10&  7\\
       -& 8&  2&  1&  4\\
\hline
   \end{tabular} \hspace{0.2cm}
   \begin{tabular}{|c| c c c |} 
   \hline
step & to score & lowerBound & upperBound\\
\hline
  1 & $\{6, 9, 10\}$ & 4.7 &5.8 \\
   2 &  $\{2, 3\}$ & 4.7 & 3.94\\
    -  & &&\\
      - & &&\\
   - & &&\\
     - & &&\\
     - & &&\\
     - & &&\\
      - & &&\\
      - & &&\\
\hline
   \end{tabular}
   
   \label{ExampleThreshold}
\end{table}


{The algorithms are illustrated on a small toy example in Table~\ref{ExampleThreshold}. Here, a database of size $N=10$ is queried for the top-$1$ set using a model with $R=4$ components. The naive approach consists of scoring all ten items and withholding the one with the largest score. Fagin's algorithm scans the sorted lists for five steps, at that point item 5 has been encountered in each list, terminating the random access phase. All nine items seen in the random access phase are scored in the sorted access phase. The threshold algorithm terminates after two steps, when the lower bound has exceeded the upper bound. Five items are scored by this algorithm. All three methods return $\{6\}$ as the correct top-$1$ set, but they differ in the number of items that have been scored.}

In what follows we compare the computational complexity of Fagin's algorithm and the threshold algorithm from a multi-target prediction perspective. To this end, we consider the cost associated with the number of targets that have to be scored, and we derive a bound accordingly. In both algorithms, the $R$ lists of the features for the $M$ targets have to be sorted. Using conventional sorting algorithms, this can be done with a time complexity of $\mathcal{O}(RM\log M)$. If the targets remain unchanged, or are updated only slowly, this is an operation that has to be done only once. Consequently, the cost of sorting should not be included in the computational cost for computing $S_x^K$.

In Algorithm~\ref{FaginAlgorithm}, lines \ref{faginrandomstart} to \ref{faginrandomend} are concerned with finding the relevant targets to score. Suppose the lists have to be followed to a depth $D \leq M$, then this part has a time complexity of $\mathcal{O}(RD)$. The last part of Fagin's algorithm is identical to the naive algorithm, resulting in a time complexity of $\mathcal{O}(M_FR)$, because $M_F$ targets are scored (with $M_F \leq M$). For independent lists (i.e.\ the position of a given target is independent for the different lists), the number of targets $M_F$ to score is of the order $M^{\frac{R-1}{R}} K^\frac{1}{R}$. Following this simplifying assumption and again ignoring the cost of maintaining the current best top-$K$ list, we obtain a time complexity of $\mathcal{O}(RM^\frac{R-1}{R}K^\frac{1}{R})$. The time complexity of Fagin's algorithm is thus less than the time complexity of the naive algorithm, but the improvement is in practice rather small. More specifically, Fagin's algorithm will calculate the fewest scores when all components of the representation $\vec{t}(y)$ have a very strong (positive or negative) correlation. However, if those components are highly correlated, they likely share information and perhaps some effort should be done to reduce the number of components, e.g.\ by means of feature selection techniques. 


The complexity analysis of the threshold algorithm is rather simple: it only calculates $M_T$ scores with $M_T\leq M$, so its computational complexity is $\mathcal{O}(M_TR)$. It has been shown in \cite{Fagin2003} that the threshold algorithm is instance-optimal, meaning that the algorithm cannot be outperformed by any other algorithm when wild guesses are not allowed.

\begin{definition}
(Instance optimality) Let $\mathbf{A}$ be a class of algorithms, let $\mathbf{Y}$ be a class of target sets and let $\mathrm{cost}(\mathcal{A},\mathcal{Y},x)$ be the cost when running algorithm $\mathcal{A}$ on target set $\mathcal{Y}$ for query $x$. We say that an algorithm $\mathcal{B}$ is instance-optimal over $\mathbf{A}$ and $\mathbf{Y}$ if $\mathcal{B} \in \mathbf{A}$ and if for every $\mathcal{A} \in \mathbf{A}$, every $\mathcal{Y} \in \mathbf{Y}$ and every $x$ in $\mathcal{X}$ it holds that
\begin{eqnarray*}
\mathrm{cost}(\mathcal{B},\mathcal{Y},x) = \mathcal{O}(\mathrm{cost}(\mathcal{A},\mathcal{Y},x))\,.
\end{eqnarray*}  
The above equation means that there exist constants $c$ and $c'$ such that $\mathrm{cost}(\mathcal{B},\mathcal{Y},x) \leq c \times \mathrm{cost}(\mathcal{A},\mathcal{Y},x) + c'$ for every choice of $x$, $\mathcal{Y} \in \mathbf{Y}$ and $\mathcal{A} \in \mathbf{A}$.
\end{definition}

In other words, if an algorithm $\mathcal{B}$ is instance-optimal, then for any query $x$ no algorithm that obtains a lower time complexity exists. However, it can still be the case that another instance-optimal algorithm will need less computations than $\mathcal{B}$, but the difference would then be attributed to a constant time factor. This type of optimality is in fact much stronger than optimality in the average case or worst case: it holds for every query. We can easily show that the threshold algorithm is instance-optimal for the problem that is the main interest of this paper. To this end, we have to make a restriction to algorithms that do not make wild guesses. 

\begin{definition}
(Wild guess) An algorithm for solving (\ref{eq:prob_statement}) is said to make a wild guess if it computes for a given target $y$ the score $s(x,y)$ before $y$ has been observed in any of the lists $L_1,...,L_R$. 
\end{definition}

\begin{theorem}\label{instance_optimal}
Let $\mathbf{Y}$ be the class of all possible target sets $\mathcal{Y}$. Let $\mathbf{A}$ be the class of algorithms that solve for any query $x$ problem (\ref{eq:prob_statement}) in an exact manner without making wild guesses. Then the threshold algorithm is instance-optimal over $\mathbf{A}$ and $\mathbf{Y}$.  
\end{theorem}

This result follows immediately from Theorem 6.1 in \cite{Fagin2003} by transforming problem statement (\ref{eq:prob_statement}) to the original problem setting of Fagin, as discussed more formally above, and observing that this leads to searching for the maximum of a monotonic aggregation operator. Crucial in the above theorem is that we make a restriction to algorithms that do not make wild guesses. Due to lucky shots, algorithms that make wild guesses would be able to find the top-$K$ set more rapidly, but they would not outperform a deterministic algorithm for general queries $x$. As we focus in this article on deterministic algorithms, we omit further details of the comparison with stochastic methods, but the analysis becomes much more delicate if one would like to extend the above theorem beyond the class of deterministic algorithms. We refer to \cite{Fagin2003} for more details. In the next statement, we formally show that Fagin's algorithm cannot be instance-optimal. 

\begin{theorem}\label{not_instance_optimal}
Let $\mathbf{Y}$ be the class of all possible target sets $\mathcal{Y}$. Let $\mathbf{A}$ be the class of algorithms that solve for any query $x$ problem (\ref{eq:prob_statement}) in an exact manner without making wild guesses. Then Fagin's algorithm is not instance-optimal over $\mathbf{A}$ and $\mathbf{Y}$.  
\end{theorem}
\begin{proof}
To show that Fagin's algorithm is not instance-optimal, it suffices to prove that it has a larger time complexity than the threshold algorithm for one particular query $x$. Adopting the same notation as in Table~1, let us assume that $
\vec{u}(x) = (1,1)$ with $R=2$ and let us define the lists $L_1$ and $L_2$ as in Table~\ref{tab:counterexample}. We consider in this example a dataset that contains $M$ targets and we assume that the values of $t_1(y)$ increase with the indices $\{1,\ldots M\}$ of the items, while the values of $t_1(y)$ increase with the indices. With the construction of $t_1(y)$ and $t_2(y)$ as in the table, Fagin's algorithm needs $M/2$ steps to terminate, whereas the threshold algorithm only needs two steps, independent of $M$. So, the former has for this artificial dataset a time complexity of $\mathcal{O}(M)$, whereas the latter has a complexity of $\mathcal{O}(1)$. As a result, Fagin's algorithm cannot be instance-optimal.  
\qed
\end{proof}

\begin{table}[tbp]
   \centering
   \topcaption{Hypothetical dataset used to show that Fagin's algorithm is not instance-optimal for $R=2$. Top: a dataset containing $M$ targets and the scores when $\vec{u}(x) = (1,1)$. Bottom: an overview of the steps needed to find the top-$1$ set for Fagin's algorithm and the threshold algorithm. The former needs $M/2$ steps to terminate, whereas the latter only needs two steps, independent of $M$. We assume in this example that the values of $t_1(y)$ increase with the indices $\{1,\ldots M\}$ of the items, while the values of $t_1(y)$ increase with the indices. For simplicity, ties occur among targets $2,3,...,M-1$. Ties can be easily removed by constructing a more complicated example. }

   \begin{tabular}{| c | cc | c |} 
   \hline
item & $t_1(y)$ & $t_2(y)$ & $s(x, y)$\\
\hline
1 & 1.1  &  0.1  & 1.1  \\ 
2 & 0.5 & 0.5 & 1.0\\
3 & 0.5 & 0.5 & 1.0\\
\vdots & \vdots & \vdots & \vdots\\
$M-1$ & 0.5 & 0.5 & 1.0\\
$M$ & 0.1  &  1.0  & 1.1  \\ 
\hline
   \end{tabular} \\ \vspace{0.3cm}
	\begin{tabular}{|c| cc |} 
   \hline
step & $L_1$ & $L_2$ \\
\hline
\rowcolor{Gray}
      1 & 1&  $\mathbf{M}$ \\
\rowcolor{Gray}    
      2 & 2&  M-1 \\
   \rowcolor{Gray}
            3 & 3&  M-2 \\
       \rowcolor{Gray}
			      \vdots & \vdots &  \vdots \\
   \rowcolor{Gray}
            $M/2$ & $M/2$ & $M/2$ \\
            $M/2 + 1$ & $M/2 + 1$ & $M/2 - 1$ \\
						 \vdots & \vdots &  \vdots \\
						 $M$ & $\mathbf{M}$ &  1 \\
\hline
   \end{tabular} \hspace{0.2cm}
   \begin{tabular}{|c| c c c |} 
   \hline
step & to score & lowerBound & upperBound\\
\hline
  1 & $\{1,M \}$ & 1.1 &2.0 \\
   2 &  $\{2, M-1 \}$ & 1.1 & 1.0\\
    -  & &&\\
   \vdots & &&\\
     - & &&\\
		    - & &&\\
     \vdots & &&\\
      - & &&\\
\hline
   \end{tabular}
	
	\label{tab:counterexample}
	\end{table}

The above theorem confirms that datasets can be found where Fagin's algorithm will suffer from a higher time complexity than the threshold algorithm. Conversely, due to the instance-optimality of the threshold algorithm, one cannot find datasets where this threshold algorithm will exhibit a higher time complexity than Fagin's algorithm. The next theorem supports a related, but different claim.  

\begin{theorem}
For any query $x$ and for any possible target set $\mathcal{Y}$, the threshold algorithm never computes more scores $s(x,y)$ than Fagin's algorithm.   
\end{theorem}
\begin{proof}
For simplicity we give a proof by contradiction for the case $K=1$ (top-1). An extension for the case $K>1$ is then immediate. Let us assume that Fagin's algorithm terminates at depth $d$ in the lists $L_i$. Suppose also that the threshold algorithm needs to compute more scores than Fagin's algorithm, implying that it terminates at a depth larger than $d$. At depth $d$ an upper bound on the score of targets that have not been observed yet is given by Eq.\ (\ref{upperboundTA}). At depth $d$, a lower bound on the highest score is given by the target that has been observed in all lists $L_i$ when Fagin's algorithm reaches this depth. We know that at depth $d$ this lower bound exceeds the upper bound in Eq.\ (\ref{upperboundTA}). As a result, the threshold algorithm should also stop at this depth, so this is a contradiction. \qed
\end{proof}

The threshold algorithm will never compute scores for more targets than Fagin's algorithm, but this does not imply that it will always need less computations. This is due to the fact that it is characterized by some additional overhead for computing the lower and upper bound in each iteration, and for verifying whether the first bound exceeds the latter. This additional overhead does not influence the time complexity of the threshold algorithm, but it explains why constants need to be considered in the definition of instance-optimality. 
  
In contrast to Fagin's algorithm, the threshold algorithm does not require a large buffer size, as only the upper and lower bounds are needed to scan the lists. Consequently, the buffer size of the threshold algorithm is as a result bounded, while this is not the case for Fagin's algorithm. The latter will hence suffer from a too large memory consumption, making it a practically less useful method, especially for high-dimensional problems. Sorting the lists can be done offline and parallel extensions can be easily implemented~\cite{Fan2014}. If speed predominates over accuracy, the threshold algorithm can be halted before the stopping criterion is reached, so that the top-$K$ set is potentially not correct, but the running time is shorter. For example, if the threshold algorithm would be terminated in the first step of the example above, we would have obtained the correct top by scoring only three targets. This modification is known as the halted threshold algorithm~\cite{Fagin2003}. Later on in the experiments we show that this heuristic could yield quite satisfactory results. 

{The threshold algorithm can also elegantly deal with sparse data. If $\vec{u}(x)$ and $\vec{t}(y)$ are large, sparse, non-negative feature vectors (e.g. in memory-based collaborative filtering or when dealing with bag-of-word features), only the non-zero elements and corresponding pointers have to be stored. By using sparse vector-vector multiplication implementations, the calculation of a single score can be reduced. The sorted lists $L_r$ only need to contain the pointers for items where $t_r(y)$ is positive. Furthermore, one only has to consider lists corresponding to non-zero elements of $\vec{u}(x)$ (Algorithm \ref{ThresholdAlgorithm}, line \ref{sparselists}). Thus, in addition to improving the computing time for calculating a single score (which can also be done using the naive algorithm), much fewer items have to be scored. We will demonstrate this in the experimental section.}

{Despite the strong theoretical result above, we show that it is still possible to slightly improve the threshold algorithm in some cases. The reason is that Theorem~\ref{instance_optimal} does not take the dimensionality $R$ of feature vectors into account.  We propose a small modification so that not all scores have to be calculated completely. Let us assume that the following double inequality holds for the score of a newly-observed target $y$ at depth $d$ of the threshold algorithm:
\begin{equation}
s(x,y) \leq \sum_{r=1}^{l} u_r(x) \, t_r(y) + \sum_{r=l+1}^R u_r(x) \, t_r(y_{L_r(d)}) \leq \mathrm{lowerBound}(d)\,,  \label{partialineq}
\end{equation}
with $1\leq l < R$. For such a target it is not needed to calculate the score entirely. Algorithm~\ref{PartialThresholdAlgorithm} modifies the scoring functionality of the original threshold algorithm by applying this idea. We will refer to this modification as the partial threshold algorithm. It stops the calculation when it becomes clear that the score cannot improve the lower bound. We start from the upper bound and gradually update the score to either obtain the full score or to halt when the partially calculated score is lower than the lower bound. The partial threshold algorithm will return the same top-$K$ set and is still an exact algorithm. The number of items that are considered is the same as with the threshold algorithm, only the calculation of some scores that do not improve the top will be halted early, resulting in a decrease in computations.
}

\begin{algorithm}[t]
\caption{Calculating the scores for the Partial Threshold Algorithm}\label{PartialThresholdAlgorithm}
\begin{algorithmic}[1]
\Statex \textbf{Input}: $\mathbf{u}(x)$, $\mathbf{t}(y)$, upperBound, lowerBound, $L_1,\ldots,L_R$, $d$
\Statex \textbf{Output}: score or fail
\Statex 

\State score $\leftarrow$ upperBound
\For{$r\leftarrow1$ \textbf{to} $R$}
	\State score $\leftarrow$ score -  $u_r(x) \, t_r(y_{L_r(d)})$
	\State score $\leftarrow$ score +  $u_r(x) \, t_r(y)$
	\If{ score $\leq$ lowerBound}
	\State \textbf{break} and return a fail \Comment{item $y$ will not improve $S^K_x$}
	\EndIf
\EndFor

\end{algorithmic}
\end{algorithm}

\section{Application domains and relationships with SEP-LR models}\label{SLRapplications}
{
In this section we illustrate that many multi-target prediction methods are specific instantiations of SEP-LR models and the inference methods discussed in the previous sections. The methods that we consider in Section~\ref{ColFilt} are frequently encountered in the area of recommender systems. The methods that we discuss in Sections~\ref{MultiOutputModels} and~\ref{pairwise_network} are used in other domains as well.

{For the application domains discussed below, many of the currently used techniques boil down to SEP-LR models. When that is the case, the algorithms of Section~\ref{queryAlgorithms} can be applied to decrease the running time at no cost in predictive performance, as we can guarantee that the correct top-$K$ predictions will be returned. When a non-SEP-LR model is used, it may be considered to switch to a SEP-LR model to perform queries potentially faster. This may lead to a decreased performance as one is restricted to linear models. Depending on the application this trade-off may be worth considering or not.}
}
\subsection{Memory-based and model-based collaborative filtering}\label{ColFilt}

In this section we discuss how our methodology can be applied to so-called user-based and item-based collaborative filtering. Since these methods directly process the available dataset, they are often denoted as memory-based collaborative filtering. To keep the discussion accessible, let us focus on the latter of the two settings, where recommendations are made by retrieving items that are similar to the items a user has seen before -- see e.g.~\cite{Sarwar2001}. A popular similarity measure in this area is the cosine similarity:
\begin{equation}
\cos(x,y) = \frac{\langle \vec{x}, \vec{y}\rangle}{\| \vec{x} \|_2  \| \vec{y} \|_2}\, , \label{cosine_similarity}
\end{equation}
in which $x$ should be interpreted as an item that has been seen before, whereas $y$ rather refers to an item that could be recommended. The vectors $
\vec{x}$ and $\vec{y}$ here typically consist of ratings or zero-one purchase flags for different users. Using this notation, item-based collaborative filtering can be interpreted as a specific case of model (\ref{s_in_tu}) by defining 
\begin{equation}
\vec{u}(x) = \frac{\vec{x}}{\| \vec{x} \|_2}\, \quad \mbox{and} \quad \vec{t}(y) = \frac{\vec{y}}{\| \vec{y} \|_2}\,.
\end{equation}
A similar reasoning can be followed for user-based collaborative filtering methods and other similarity measures such as the Pearson correlation and the adjusted cosine similarity. Both measures can also be written as dot products, but sparsity will be lost due to centring.

{In contrast to the memory-based approach, collaborative filtering can also be performed by building a model to explain the preference of a user for a particular target. This is referred to as model-based collaborative filtering and includes models such as Bayesian networks, clustering and matrix factorization methods~\cite{Lee2012}. Here, each target and each user are associated with a vector of $R$ latent variables of which the dot product represents the joint preference. A matrix $\mathbf{C}$ containing the ratings for each combination of objects $x$ and $y$ is approximated as a product of two low-rank matrices:
$\mathbf{C} \approx \mathbf{U} \mathbf{T}$. This means that the predicted scores in model (\ref{s_in_tu}) can be written in matrixform as $\mathbf{S} = \mathbf{U} \mathbf{T}$, such that the score $s(x,y)$ of the $i$-th user and $j$-th item is computed as a dot-product of the $i$-th row in $\mathbf{U}$ and the $j$-th column in $\mathbf{T}$.  As a result, the connection with SEP-LR models holds for basically all matrix decomposition algorithms, inclucing non-negative matrix factorization, independent component analysis, sparse principal component analysis and singular value decomposition methods~\cite{hastie01statisticallearning}. 
}

\subsection{Multi-label classification, multivariate regression and multi-task learning}\label{MultiOutputModels}

In this section we discuss applications of SEP-LR models in three large subfields of machine learning that have a lot of commonalities, namely multi-label classification, multivariate regression and multi-task learning. The former two subfields consider the simultaneous prediction of multiple binary or real-valued targets, respectively~\cite{Dembczynski2012,Tsoumakas2007}. Multi-task learning then further unifies those subfields, and further extends them to problems where not all targets are observed (or relevant) for all instances~\cite{Ben-David2003,Caruana1997}. However, all multi-label classification and multivariate regression problems can also be solved with multi-task learning methods. {Of course, finding the top-$K$ set is only relevant when an ordering of the targets is appropriate. For example, consider a protein-ligand interaction prediction problem where the goal is to find for a given ligand a set of proteins from a large database that are likely to show affinity towards it. If the model that is used predicts binding energy, the problem is a multivariate regression problem. When the model returns a value that is related to the probability that a given ligand will interact with a particular protein or not, we are in a multi-label classification setting.} 

Using the notation introduced before, this means that $x$ represents an instance and $y$ a target (a.k.a.\ label or output). As we will be computing the $K$ most relevant targets for a given instance $x$, we end up with a so-called label ranking setting. Furthermore, suppose we have $M$ targets and a set of $N_{\text{tr}}$ training instances for which a feature representation $\boldsymbol{\psi}(x)$ is available.  

Many multi-label classification, multivariate regression and multi-task learning methods can be formulated as SEP-LR models. Due to lack of space, we only discuss a few representative cases. The first approach is a very simple approach: it defines independent models for different targets. This method is known as the binary relevance method in the multi-label classification community~\cite{Tsoumakas2007}. Using linear models or linear basis function models, this implies that the model for target $y$ can be represented as 
$$s(x,y) = \vec{w}_y^\intercal \boldsymbol{\psi}(x) \,.$$
Hence, this model is a specific case of the SEP-LR framework, with $\vec{u}(x) =\boldsymbol{\psi}(x)$ and $\vec{t}(y) = \vec{w}_y$.

Multivariate (kernel) ridge regression is probably the most basic model of this kind. {Ridge regression constructs a model for every target separately, ignoring potential dependencies between targets. Such dependencies can be modelled using more sophisticated methods. In multivariate regression, dependencies between targets are often modelled using (kernel) PCA, which transforms the problem into a set of uncorrelated tasks (see e.g.~\cite{Weston2006}). An alternative but related approach consists of using a specialised loss function or regularisation term that enforces models of different targets to behave more similar, e.g.~\cite{Evgeniou2005,Evgeniou2004,Jalali2010}. An interesting example from this point of view is the partial least squares (PLS) algorithm, which projects the features onto a lower-dimensional subspace~\cite{Shawe-Taylor2004}. The inference algorithms discussed in this work are applicable to all these methods.
}

\subsection{Pairwise learning, content-based filtering and supervised network inference}\label{pairwise_network}

Finally, there are applications where feature representations are available for both $x$ and $y$, resulting in separable linear relational models of the following type:

\begin{equation*}
s(x,y) = \vectorization(\mathbf{W})^\intercal(\boldsymbol{\phi}(y)\otimes \boldsymbol{\psi}(x)) = \boldsymbol{\psi}(x)^\intercal \mathbf{W}\boldsymbol{\phi}(y) \label{relationalVecMat} \,.
\end{equation*}

Here, $\vectorization$ is the vectorisation operator, which stacks the columns of a matrix into a vector and $\otimes$ is the Kronecker product. The second part of the equation is used to stress that this can be seen as a standard linear model to make predictions for a pair of objects. This model can be cast in the form of Eq.\ (\ref{s_in_tu}) by writing the terms explicitly:
\begin{equation*}
s(x,y) = \sum_{i,j} w_{ij}  \boldsymbol{\psi}(x)_i\boldsymbol{\phi}(y)_j\,,
\end{equation*}
with $w_{ij}$ denoting the element from $\mathbf{W}$ on the $i$-th row and $j$-th column.

Models of this type are frequently encountered in several subfields of machine learning, such as pairwise learning, content-based filtering, biological network inference, dyadic prediction and conditional ranking. They are particularly popular in certain domains of bioinformatics, such as the prediction of protein-protein interactions~\cite{Ben-Hur2005,Hue2010,Vert2007}, enzyme function prediction~\cite{Stock2014} and proteochemometrics~\cite{Gonen2012,Jacob2008a,Yamanishi2010}.

Some authors, such as~\cite{Ding2013}, claim that pairwise methods in kernel form are not viable for large-scale problems as they scale $\mathcal{O}(N^2M^2)$ in memory and $\mathcal{O}(N^3M^3)$ in time complexity for most algorithms. For some models though, when the kernels can be factorised using the Kronecker product, the memory and time complexity reduce to $\mathcal{O}(N^2+M^2)$ and $\mathcal{O}(N^3+M^3)$, respectively, see~\cite{pahikkala2013conditional,Waegeman2012}. In a similar vein, this factorization allows for fast querying using the threshold algorithm and Fagin's algorithm. 

Recently, a method called two-step regularised least squares was proposed, as an intuitive way to bring this paradigm to practice~\cite{Pahikkala2014}. Prior knowledge on similarities of the different tasks is incorporated in a standard multi-output ridge regression by performing a second ridge regression, using a kernel matrix describing the similarity between these tasks. It can be shown that two-step regularised least squares is equivalent to using ridge regression with a kernel based on the Kronecker product.

\section{Experimental results}\label{experiments}

The goal of this experimental section is twofold. Firstly, we want to show that the proposed methods can be applied to the different machine learning settings described above. In addition to that, we also want to illustrate the scalability and practical use of the algorithms. By means of several case studies using different datasets we will elucidate the algorithms of Section~\ref{queryAlgorithms}. Fagin's algorithm is included in this manuscript for didactic interest, but we observed in initial experiments that it could not cope adequately with most of the higher-dimensional problems. The buffer needed to store targets that are encountered in the lists grows rather quickly when the dimensionality of the data increases. Hence, we have omitted experiments with this algorithm.  Instead, for the experimental validation we focus on the practical use of the threshold and partial threshold algorithm and compare their performance with the naive algorithm.

\subsection{Data dependencies in collaborative filtering}\label{CFexperiments}

In this series of experiments we study the scalability of the threshold algorithm for a collection of collaborative filtering datasets. We use five sparse datasets for memory-based and model-based collaborative filtering using the cosine similarity and matrix factorization respectively. The datasets are listed in Table~\ref{ColFiltData}. We use the Movielens dataset with both 100K and 1M, the BookCrossing dataset, both from Grouplens\footnote{http://grouplens.org/}, the Audioscrobbler dataset\footnote{\href{http://www-etud.iro.umontreal.ca/~bergstrj/audioscrobbler_data.html}{http://www-etud.iro.umontreal.ca/\~{}bergstrj/audioscrobbler\_{}data.html}}, and a recipe composition dataset\footnote{This dataset \cite{Ahn2011} is strictly speaking not a collaborative filtering benchmark dataset, but it can be treated using a similar workflow. Here the goal would be to find related recipes for a particular ingredient combination. See~\cite{DeClercq2015a} for an recommender system built using this dataset.}. We have removed empty rows and columns and taken the logarithm of the positive values of the Audioscrobbler dataset. In the Audioscrobber and Recipes datasets the scores are obtained by implicit feedback (i.e.\ a zero indicates no observation for a given item-user pair), for the others by explicit feedback (i.e.\ a zero indicates no match for a given item-user pair).

\begin{table}[tb]
   \centering
   \topcaption{Overview of the properties of the different datasets used for the collaborative filtering experiments using memory- and model-based methods.} 
   \begin{tabular}{l l l l l} 
\hline
\hline
Name & \# rows & \# colums & \# non-zero elements& Feedback \\ 
\hline
Audioscrobbler & 73,458 & 47,085 & 656,632 & implicit\\
BookCrossing & 105,283 & 340,538 & 1,149,780 & explicit\\
Movielens100K & 943 & 1,682 & 100,000 & explicit\\
Movielens1M &  6,040 & 3,952 & 1,000,000 & explicit\\
Recipes & 56,498 & 381 & 464,407 & implicit\\
\hline
\hline
   \end{tabular}
   \label{ColFiltData}
\end{table}

{
For memory-based collaborative filtering, we used the cosine similarity to define a set of items with a high similarity to the query. As described above, if each item is normalized, such that the $L_2$-norm is equal to one, a dot product is equivalent to the cosine similarity. For model-based collaborative filtering, we applied matrix factorization by means of  probabilistic PCA~\cite{Tipping1997}, having as advantage that a computationally efficient expectation maximisation algorithm can be used. We used 5, 10, 50, 100 and 250 latent features in the decomposition of all datasets. 
}

{
We performed a large series of queries to compare the performance of the different algorithms discussed in Section\ \ref{queryAlgorithms}. For memory-based collaborative filtering, we used the naive and threshold algorithms. For matrix factorization, the partial threshold algorithm was also used for querying. For each experiment we randomly chose a dataset, the size of the top among the values 1, 5, 10, 50 and 100 and ten queries from the rows of the data matrix. To assess the effect of the size of the database, we randomly witheld 10\%, 50\% or all the items in the database. Each of the algorithms was applied to find the same top-$K$ relevant targets for all these settings.
}

{
The number of scores calculated by the threshold algorithm relative to the number of scores calculated by the naive algorithm for the different settings is given in Figure~\ref{colfiltresult}. As guaranteed by the design of the threshold algorithm, in the worst-case scenario the same number of items has to be scored as by the naive algorithm. Clearly, the average gain increases for both settings with database size, as expected. Furthermore, the larger the top-$K$ that is sought, the smaller the gain. The difference between top-1 and top-50 seems to be roughly an order of magnitude for most cases. Comparing the memory-based with the model-based setting, the former shows larger improvements in computational costs, even for small datasets and large top sizes. This is because the threshold algorithm can deal well with sparseness, as discussed in Section~\ref{queryAlgorithms}. For the matrix factorization experiments, we see that the relative efficiency decreases when the number of latent features $R$ increases. Large improvements in running time are mainly to be expected when the dimensionally of the data is modest.
}

\begin{figure}[htbp]
   \centering
   \includegraphics[width=\textwidth]{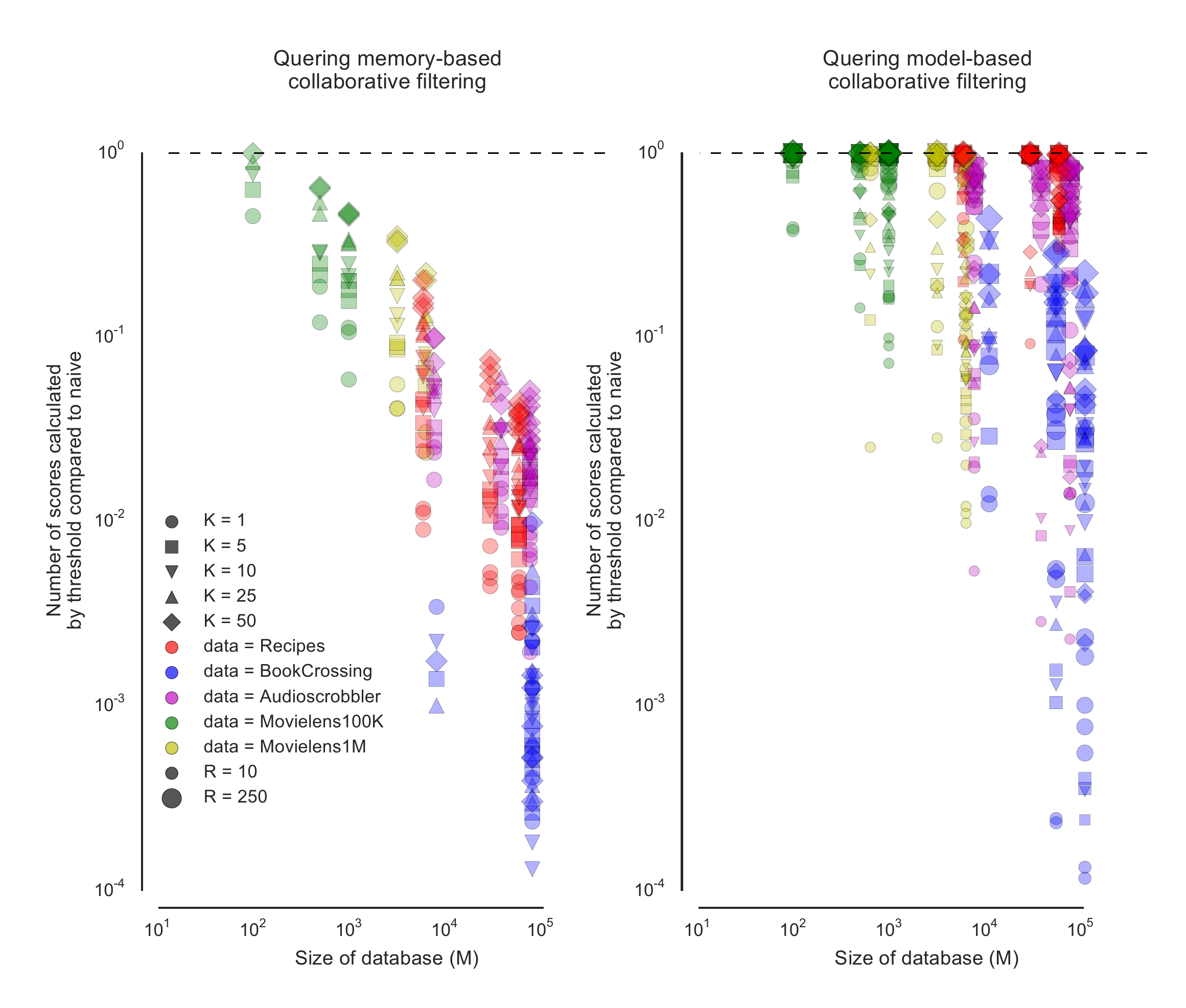} 
   \caption{Performance of the threshold algorithm relative to the naive algorithm for the collaborative filtering models based on memory and matrix factorization. Marker shape, color and sizes indicates top size, dataset and number of latent features in the case of model-based collaborative filtering.}
   \label{colfiltresult}
\end{figure}

\subsection{Running time and partial threshold algorithm using Uniprot data}

{
In this experiment we study the improvement in querying time and whether the computation time can be reduced by means of the partial threshold algorithm. To this end we use a large multi-label classification dataset related to protein function. We downloaded the complete Uniprot database\footnote{http://www.uniprot.org}, containing more than half a million of annotated protein sequences. We consider the problem of functionally annotating a protein based on its amino acid sequence. The function of the protein is represented by the gene ontology (GO) label. The GO annotation represents the molecular function, location and biological role of a protein. We ended up with a dataset containing 211,149 proteins with 21,274 distinct labels.
}

As a feature representation, we used weighted subsequence kernels~\cite{Shawe-Taylor2004}, which define similarity between sequences based on subsequences of a fixed length. It is infeasible to compute a complete kernel matrix for hundreds of thousands of instances. We performed an approximation inspired by the Nystr\"om method~\cite{Drineas2005}, a method to approximate the decomposition of large kernel matrices. We arbitrarily selected 500 reference substrings with a length of 100 amino acids from the database, for which we calculated the different subsequence kernels with subsequence lengths from one to twenty with respect to the full protein sequences in the dataset. Thus, if we treat the subsequence length as a hyperparameter, each protein is described by 500 features. In order to find the best model, the models were trained using 80\% of the data, while 20\% was withheld as a validation set. All models were evaluated by calculating the AUC over the different functional labels, for each protein separately, and averaging over different proteins. Apart from the subsequence length, each model has its own hyperparameter that had to be optimized:
\begin{itemize}
\item \textbf{Ridge}: the best combination of substring length and regularisation parameter from the grid $\{0.01, 0.1, 1, 10, 100\}$ was chosen,
\item \textbf{PLS}: the dimension of the latent subspace was treated as a hyperparameter, and selected from the grid $\{10, 50, 100, 250\}$.
\end{itemize}
{
We measured the classification performance using the area under the ROC curve (AUC) over the labels, averaged over all proteins in the test set (i.e. instance-wise AUC). For ridge regression, we obtained an AUC of 0.982, while for the PLS model an AUC of 0.980 was attained. We used the naive, threshold and partial threshold algorithm to query for the top 1, 5, 10, 25 or 50 most likely labels on the complete label space or random 10\% or 50\% subsets. Each datapoint represents the average efficiency of a batch of ten protein queries. The results are represented in Figure~\ref{UniprotResults}.
}

{
In the first graph of Figure~\ref{UniprotResults}, we plotted the relative improvement in number of scores of the threshold algorithm plotted against the improvement in running time for a series of queries on the Uniprot dataset. It is clear that an improvement in the number of scores translates into a proportional improvement in running time ($R^2=0.964$). Similarly as for the collaborative filtering results, larger improvements occur for large databases and small top sizes. It is interesting to note that much larger improvements are recorded for the ridge regression model than for PLS. Most likely, this is because PLS orthogonalizes the variables before mapping to the output space, influencing the distribution of the features. 
} 

{
The second graph of Figure~\ref{UniprotResults} compares the computational performance of the partial threshold algorithm described in Section~\ref{queryAlgorithms} with the standard threshold algorithm. Recall that the number of items the partial threshold algorithm tries to score is always the same as for the threshold algorithm. In most cases, the partial threshold algorithm will only calculate a fraction of a score. To compare with the threshold algorithm, we measure the average fraction of partial scores $u(x)_rt(y)_r$ that are added. Hence, each item that is fully scored by the partial threshold algorithm is given a value of one, while partially scored items receive the value of the proportion calculated in Algorithm~\ref{PartialThresholdAlgorithm}. Here, we show that the partial threshold algorithm always has to calculate less full scores compared to the threshold algorithm. Since this algorithm has some extra overhead compared to the original threshold algorithm (the score is updated in two steps), this algorithm is only expected to improve in time when on average only a relatively small fraction of a score has to be calculated. Only for ridge regression with top size equal to one, we see an improvement noticeably larger than 200\%. For these queries we observe a small improvement in running time compared to the threshold algorithm, while for the other experiments a small increase in running time is detected. Likewise, for the model-based collaborative filtering datasets of Section~\ref{CFexperiments}, a similar pattern was observed. The partial threshold algorithm has to calculate only roughly half a score on average and in practice shows something between a minor increase or decrease in running time.
}

\begin{figure}[t] 
   \centering
   \includegraphics[width=\textwidth]{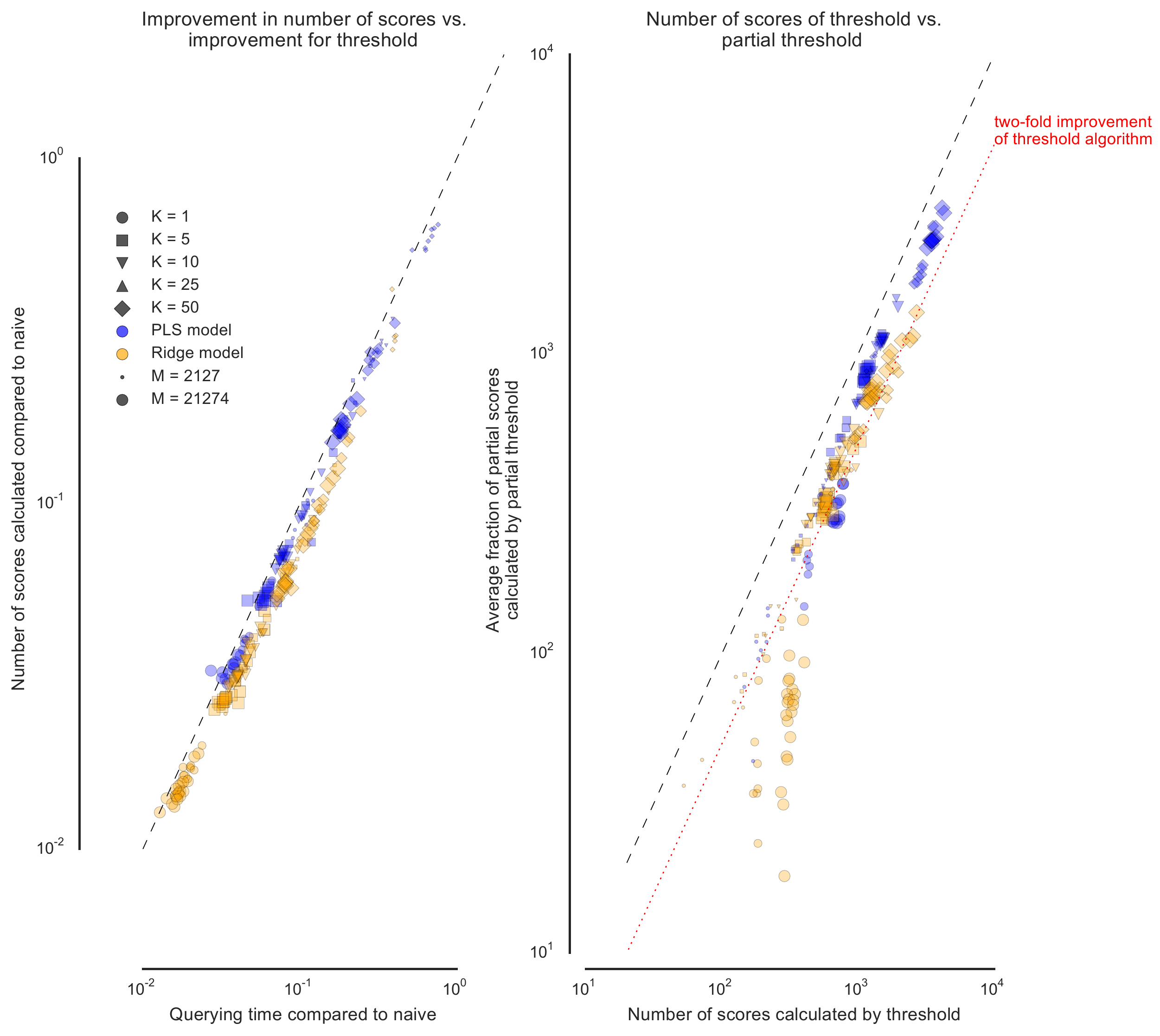} 
   \caption{Left: improvement in number of scores of the threshold algorithm plotted against the improvement in running time for a series of queries on the Uniprot dataset. Right: number of scores that are calculated by the partial threshold algorithm versus the number of scores needed for the threshold algorithm to find the top labels for the Uniprot data.}
   \label{UniprotResults}
\end{figure}

\subsection{Behaviour of individual queries}\label{queryzoom}

\begin{figure}[htbp]
   \centering
   \includegraphics[width=\textwidth]{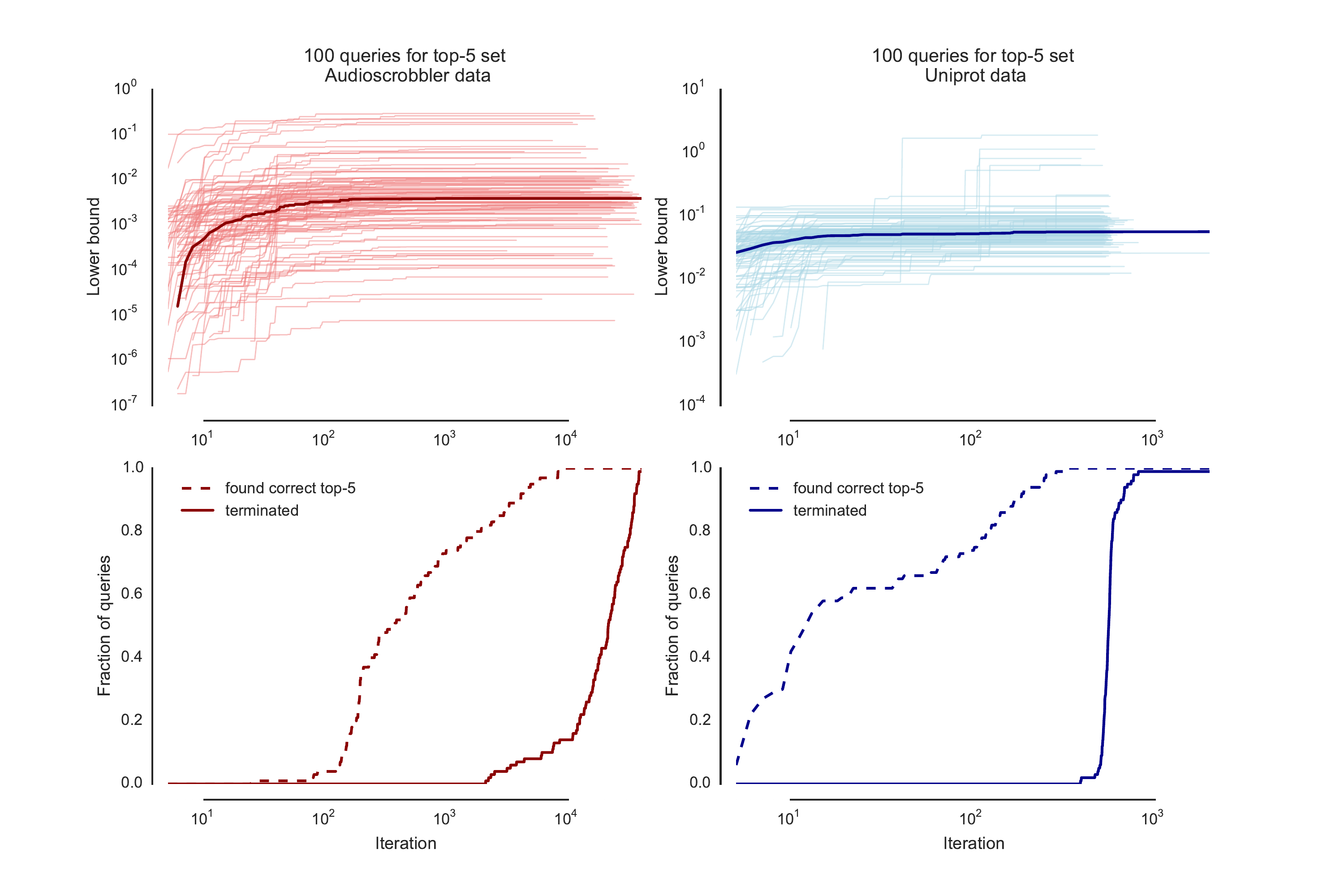} 
   \caption{Detail of 100 queries for the top-5 for the Audioscrobbler dataset (model-based collaborative filtering with 50 latent features) and the Uniprot dataset (using partial least squares). The top graphs give the values of the lower bound for the individual queries throughout the run of the threshold algorithm. The median of the upper bound is given by the bold line. The bottom graphs show the fraction of queries that have found the correct top at a step in the algorithm and after how many steps the algorithm terminates.}
   \label{induvidualqueryplot}
\end{figure}

{
To study the behaviour of the threshold algorithm in more detail, we analyzed the lower bound for 100 queries for the top-5 for the Audioscrobbler dataset (model-based collaborative filtering with 50 latent features) and the Uniprot dataset (using partial least squares). The results are represented in Figure~\ref{induvidualqueryplot}. These plots demonstrate that even though it may take many iterations for the algorithm to terminate, the correct top is often found swiftly. The top graphs monitor the lower bound (i.e.\ the worst score in the current top-$K$ set). In most cases this lower bound does not increase dramatically after a few hundred iterations. Hence, if the desired top does only need to contain good candidates instead of the correct top-$K$ set, one can decrease computation time by early halting. Furthermore, the lower plots demonstrate that, at least for these cases, the threshold algorithm often finds the correct top much sooner than it terminates. This can be seen in the lower plots of Figure~\ref{induvidualqueryplot} where there is a strong lag between the time the correct top is found and the time the algorithm returns this top. We can conclude that it could make sense in some cases to use the halted threshold algorithm (see~\cite{Fagin2003}) which stops when the computational resources have run out (e.g.\ when a specified maximum querying time has been exceeded).
}
\subsection{Large-scale text classification}

In this final experiment we demonstrate the scalability of the threshold algorithm on a huge text mining task with the goal of putting Wikipedia articles in 325,056 possible categories. This dataset was downloaded from the Kaggle `Large Scale Hierarchical Text Classification' challenge\footnote{https://www.kaggle.com/c/lshtc} \cite{Partalas2015}. The training data contained 2,365,436 articles which were labeled with the different categories they belong to. Each article had a feature description in the form of a sparse bag-of-words vector, with a length of 1,617,899. The word counts were transformed into term frequency-inverse document frequencies~\cite{Manning2008}. We performed partial least squares regression to map the feature vectors to the output space. This approach is similar to that performed by~\cite{Mineiro2014} on this dataset. The models had 10, 50, 100, 500 and 1000 latent features. The threshold algorithm was used for the different models to fetch the highest-scoring class (i.e. $K$=1). The models were evaluated using average AUC over the labels and precision-at-1 (the percentage of queries that returned the correct class at the top) on a random test set of 10,000 articles. As a baseline, we also included a popular recommender which always returns a score proportional to the probability of the label occurrence, independent of the word-features. The results are represented in Table~\ref{WikiPerf}.

\begin{table}[tb]
   \centering
   \topcaption{Performance of the models for the text classification data. The partial least squares models were calculated with different numbers of latent features $R$. The performance is calculated for 10,000 test articles of the categories and averaged. A popular recommender is used as a baseline and has an AUC of 0.848 and a precision-at-1 of 0.165.} 
   \begin{tabular}{l  lllll} 
\hline
\hline
$R$ & 10 & 50 & 100 & 500 & 1000 \\ 
\hline
AUC &0.918& 0.955 & 0.962& 0.975 & 0.980 \\
Precision-at-1 & 0.182 & 0.191 & 0.197 & 0.218 & 0.237\\
Average number of scores calculated & 28.3 & 179.4 & 441.7 & 3581.3 & 8995.7\\
\hline \hline
   \end{tabular}
   \label{WikiPerf}
\end{table}

All our models have a better prediction accuracy than the baseline and the performance increases with the number of latent features used. As was the case for the previous experiments, the average number of scores needed to find the top increases with $R$, but seems to scale well (when $R$ is 1000, on average only 2.8\% of the classes have to be considered).

\section{Conclusions and future perspectives}\label{Faginconclfutwork}


In this paper we have discussed three generally-applicable methods for finding efficiently the most relevant targets in a wide range of multi-target prediction problems. The experimental results clearly confirm that the threshold algorithm is quite scalable, and often many orders of magnitude more efficient than the naive approach. We can thus conclude that this algorithm is very suitable for querying in large-data settings using many popular machine learning techniques. {We have shown experimentally that large improvements can be expected when the size $M$ of the database is large and when only a small top-$K$ set is requested. Furthermore, the number of scores that have to be calculated increases with the dimension $R$. When querying time is vital, it is thus recommended to consider the use of (supervised) dimension reduction techniques on the features. Empirically, we have also shown that the threshold algorithm can cope well with large sparse feature vectors.}

{
We proposed a small modification to the threshold algorithm. By making use of the lower and upper bound, the calculation of the scores can be terminated early while still retaining an exact method. In the experiments, we provided a proof-of-concept, showing a modest reduction in the number of calculations. In practice, due to the extra computational overhead that is created, this is not always translated into an reduced running time. We are optimistic that this modification can be of use in future applications for two reasons. Firstly, certain hardware implementations might be suited for this algorithm as it is now hard to beat optimized implementations for ordinary matrix multiplication. Secondly, we believe that our partial threshold algorithm may lead to very efficient inexact querying algorithms. In Section~\ref{queryzoom} we demonstrated that the threshold algorithm is often overly conservative to determine when the correct top is found. In future work, we will study the trade-off between uncertainty in the top-$K$ set and  computational cost.
}

The framework that we presented works well for SEP-LR models corresponding to Eq.\ (\ref{s_in_tu}), which applies to a wide range of statistical models. However, in certain applications, one also encounters pairwise features that jointly describe information of a query and a target. For example in bioinformatics, if one wants to find miRNAs (small nucleic acid sequences with a regulating function) that can interact with messenger RNA, one sometimes considers the number of nucleotide matches between these sequences (e.g.~\cite{DePaepe2015}). Similarly, when recommending products to users, one could consider as a feature the number of views a user had on a page where the product is mentioned. Features of that kind cannot be incorporated into SEP-LR models. One would need linear models of the following form:
\begin{equation*}
s(x,y) = \sum_{r = 1}^R q_r(x, y)\, ,
\end{equation*}
with $\boldsymbol{q}(x,y) = (q_1(x, y),q_2(x, y),...,q_R(x, y))^\intercal$ a pairwise representation of the query and target. This can be thought of as a generalisation of Eq.\ (\ref{s_in_tu}). Note that the feature description is never known in advance, because it depends upon the query $x$. Finding efficient algorithms for dealing with these types of models is also an interesting topic for future research.
\section*{Acknowledgments}
Part of this work was carried out using the Stevin Supercomputer Infrastructure at Ghent University, funded by Ghent University, the Hercules Foundation and the Flemish Government - department EWI.
\vskip 0.2in
\bibliographystyle{spmpsci}      

\bibliography{library,referenties}

\begin{thebibliography}{10}
\providecommand{\url}[1]{{#1}}
\providecommand{\urlprefix}{URL }
\expandafter\ifx\csname urlstyle\endcsname\relax
  \providecommand{\doi}[1]{DOI~\discretionary{}{}{}#1}\else
  \providecommand{\doi}{DOI~\discretionary{}{}{}\begingroup
  \urlstyle{rm}\Url}\fi

\bibitem{Agarwal2012}
Agarwal, D., Gurevich, M.: {Fast top-k retrieval for model based
  recommendation}.
\newblock In: Proceedings of the Fifth ACM International Conference on Web
  Search and Data Mining, pp. 483--492
\newblock
\newblock  (2012)

\bibitem{Agrawal2013}
Agrawal, R., Gupta, A., Prabhu, Y., Varma, M.: {Multi-label learning with
  millions of labels: Recommending advertiser bid phrases for web pages}.
\newblock In: Proceedings of the 22nd International Conference on World Wide
  Web, pp. 13--24
\newblock
\newblock  (2013)

\bibitem{Ahn2011}
Ahn, Y.Y., Ahnert, S.E., Bagrow, J.P., Barab{\'{a}}si, A.L.: {Flavor network
  and the principles of food pairing}.
\newblock Scientific Reports \textbf{1}(196)
\newblock
\newblock  (2011)

\bibitem{Basilico2004a}
Basilico, J., Hofmann, T.: {Unifying collaborative and content-based
  filtering}.
\newblock In: Proceedings of the 21st International Conference on Machine
  Learning, pp. 9--16
\newblock
\newblock  (2004)

\bibitem{Ben-David2003}
Ben-David, S., Schuller, R.: {Exploiting task relatedness for multiple task
  learning}.
\newblock In: Proceedings of the 16th Annual Conference on Computational
  Learning Theory and 7th Kernel Workshop, pp. 567--580
\newblock
\newblock  (2003)

\bibitem{Ben-Hur2005}
Ben-Hur, A., Noble, W.S.: {Kernel methods for predicting protein-protein
  interactions}.
\newblock Bioinformatics \textbf{21}(Suppl 1), i38--46
\newblock
\newblock  (2005)

\bibitem{Bentley1975}
Bentley, J.L.: {Multidimensional binary search trees used for associative
  searching}.
\newblock Communications of the ACM \textbf{18}, 509--517
\newblock
\newblock  (1975)

\bibitem{Beygelzimer2006}
Beygelzimer, A., Kakade, S., Langford, J.: {Cover trees for nearest neighbor}.
\newblock In: Proceedings of the 23rd International Conference on Machine
  Learning, pp. 97--104
\newblock
\newblock  (2006)

\bibitem{Blockeel1998}
Blockeel, H., {De Raedt}, L., Ramon, J.: {Top-down induction of clustering
  trees}.
\newblock In: Proceedings of the Fifteenth International Conference on Machine
  Learning, pp. 55--63
\newblock
\newblock  (1998)

\bibitem{Caruana1997}
Caruana, R.: {Multitask learning}.
\newblock Machine learning \textbf{75}, 41--75
\newblock
\newblock  (1997)

\bibitem{chu2009}
Chu, W., Park, S.T.: {Personalized recommendation on dynamic content using
  predictive bilinear models}.
\newblock In: Proceedings of the 18th International Conference on World Wide
  Web, pp. 691--700
\newblock
\newblock  (2009)

\bibitem{DeClercq2015a}
{De Clercq}, M., Stock, M., {De Baets}, B., Waegeman, W.: {Data-driven recipe
  completion using machine learning methods}.
\newblock Trends in Food Science {\&} Technology \textbf{49}, 1--13
\newblock
\newblock  (2015)

\bibitem{DePaepe2015}
{De Paepe}, A., {Van Peer}, G., Stock, M., Volders, P.J., Vandesompele, J., {De
  Baets}, B., Waegeman, W.: {miRNA target prediction through modeling
  quantitative and qualitative miRNA binding site information in a stacked
  model structure}.
\newblock Nucleic Acid Research \textbf{Submitted}
\newblock
\newblock  (2015)

\bibitem{Dembczynski2012}
Dembczynski, K., Waegeman, W., Cheng, W., H{\"{u}}llermeier, E.: {On label
  dependence and loss minimization in multi-label classification}.
\newblock Machine Learning \textbf{88}(1-2), 5--45
\newblock
\newblock  (2012)

\bibitem{Ding2013}
Ding, H., Takigawa, I., Mamitsuka, H., Zhu, S.: {Similarity-based machine
  learning methods for predicting drug-target interactions: a brief review}.
\newblock Briefings in Bioinformatics \textbf{14}(5), 734--47
\newblock
\newblock  (2013)

\bibitem{Drineas2005}
Drineas, P., Mahoney, M.: {On the Nystr{\"{o}}m method for approximating a Gram
  matrix for improved kernel-based learning}.
\newblock Journal of Machine Learning Research \textbf{6}, 2153--2175
\newblock
\newblock  (2005)

\bibitem{Elkan2003}
Elkan, C.: {Using the triangle inequality to accelerate k-means}.
\newblock In: Proceedings of the 20th International Conference on Machine
  Learning, pp. 147--153
\newblock
\newblock  (2003)

\bibitem{Evgeniou2005}
Evgeniou, T.: {Learning multiple tasks with kernel methods}.
\newblock Journal of Machine Learning Research \textbf{6}, 615--637
\newblock
\newblock  (2005)

\bibitem{Evgeniou2004}
Evgeniou, T., Pontil, M.: {Regularized multi-task learning}.
\newblock In: Proceedings of the Tenth ACM SIGKDD International Conference on
  Knowledge Discovery and Data Mining, pp. 109--117
\newblock
\newblock  (2004)

\bibitem{Fagin1999}
Fagin, R.: {Combining fuzzy information from multiple systems}.
\newblock Journal of Computer and System Sciences \textbf{58}(1), 83--99
\newblock
\newblock  (1999)

\bibitem{Fagin2003}
Fagin, R., Lotem, A., Naor, M.: {Optimal aggregation algorithms for
  middleware}.
\newblock Journal of Computer and System Sciences \textbf{66}(4), 614--656
\newblock
\newblock  (2003)

\bibitem{Fan2014}
Fan, W., Huai, J.P.: {Querying big data: bridging theory and practice}.
\newblock Journal of Computer Science and Technology \textbf{29}(5), 849--869
\newblock
\newblock  (2014)

\bibitem{Goel2009}
Goel, S., Langford, J., Strehl, A.: {Predictive indexing for fast search}.
\newblock In: Advances in Neural Information Processing Systems, pp. 505--512
\newblock
\newblock  (2009)

\bibitem{Gonen2012}
G{\"{o}}nen, M.: {Predicting drug-target interactions from chemical and genomic
  kernels using Bayesian matrix factorization}.
\newblock Bioinformatics \textbf{28}(18), 2304--10
\newblock
\newblock  (2012)

\bibitem{hastie01statisticallearning}
Hastie, T., Tibshirani, R., Friedman, J.: {The Elements of Statistical
  Learning}.
\newblock Springer Series in Statistics. Springer New York Inc., New York, NY,
  USA
\newblock
\newblock  (2001)

\bibitem{Hue2010}
Hue, M., Riffle, M., Vert, J.P., Noble, W.S.: {Large-scale prediction of
  protein-protein interactions from structures}.
\newblock BMC Bioinformatics \textbf{11}(144), 1--10
\newblock
\newblock  (2010)

\bibitem{Ilyas2008}
Ilyas, I.F., Beskales, G., Soliman, M.A.: {A survey of top-k query processing
  techniques in relational database systems}.
\newblock ACM Computing Surveys \textbf{40}(4)
\newblock
\newblock  (2008)

\bibitem{Jacob2008a}
Jacob, L., Hoffmann, B., Stoven, V., Vert, J.P.: {Virtual screening of GPCRs:
  an in silico chemogenomics approach}.
\newblock BMC Bioinformatics \textbf{9}(1), 1--16
\newblock
\newblock  (2008)

\bibitem{Jacob2008}
Jacob, L., Vert, J.P.: {Protein-ligand interaction prediction: an improved
  chemogenomics approach}.
\newblock Bioinformatics \textbf{24}(19), 2149--2156
\newblock
\newblock  (2008)

\bibitem{Jalali2010}
Jalali, A., Sanghavi, S., Ravikumar, P., Ruan, C.: {A dirty model for
  multi-task learning}.
\newblock In: Neural Information Processing Symposium, pp. 964--972
\newblock
\newblock  (2010)

\bibitem{Koenigstein2012}
Koenigstein, N., Ram, P., Shavitt, Y.: {Efficient retrieval of recommendations
  in a matrix factorization framework}.
\newblock In: Proceedings of the 21st ACM International Conference on
  Information and Knowledge Management, pp. 535--544
\newblock
\newblock  (2012)

\bibitem{Lee2012}
Lee, J., Sun, M., Lebanon, G.: {A comparative study of collaborative filtering
  algorithms}.
\newblock ACM Transactions on the Web \textbf{5}(1), 1--27
\newblock
\newblock  (2011)

\bibitem{Manning2008}
Manning, C.D., Raghavan, P., Sch{\"{u}}tze, H.: {Introduction to Information
  Retrieval}.
\newblock Cambridge University Press, New York, NY, USA
\newblock
\newblock  (2008)

\bibitem{Mineiro2014}
Mineiro, P., Karampatziakis, N.: {Fast label embeddings for extremely large
  output spaces}.
\newblock CoRR \textbf{abs/1412.6}
\newblock
\newblock  (2014)

\bibitem{Omohundro1989}
Omohundro, S.M.: {Five balltree construction algorithms}.
\newblock Science \textbf{51}, 1--22
\newblock
\newblock  (1989)

\bibitem{pahikkala2013conditional}
Pahikkala, T., Airola, A., Stock, M., {De Baets}, B., Waegeman, W.: {Efficient
  regularized least-squares algorithms for conditional ranking on relational
  data}.
\newblock Machine Learning \textbf{93}(2-3), 321--356
\newblock
\newblock  (2013)

\bibitem{Pahikkala2014}
Pahikkala, T., Stock, M., Airola, A., Aittokallio, T., {De Baets}, B.,
  Waegeman, W.: {A two-step learning approach for solving full and almost full
  cold start problems in dyadic prediction}.
\newblock Lecture Notes in Computer Science \textbf{8725}, 517--532
\newblock
\newblock  (2014)

\bibitem{Partalas2015}
Partalas, I., Kosmopoulos, A., Baskiotis, N., Artieres, T., Paliouras, G.,
  Gaussier, E., Androutsopoulos, I., Amini, M.R., Galinari, P.: {LSHTC: a
  benchmark for large-scale text classification}.
\newblock submitted to CoRR pp. 1--9
\newblock
\newblock  (2015)

\bibitem{Sarwar2001}
Sarwar, B., Karypis, G., Konstan, J., Reidl, J.: {Item-based collaborative
  filtering recommendation algorithms}.
\newblock In: Proceedings of the Tenth International Conference on World Wide
  Web - WWW '01, pp. 285--295. ACM Press, New York, New York, USA
\newblock
\newblock  (2001)

\bibitem{Shawe-Taylor2004}
Shawe-Taylor, J., Cristianini, N.: {Kernel Methods for Pattern Analysis}.
\newblock Cambridge University Press
\newblock
\newblock  (2004)

\bibitem{Shrivastava2014}
Shrivastava, A., Li, P.: {Asymmetric lsh (ALSH) for sublinear time maximum
  inner product search (MIPS)}.
\newblock In: Advances in Neural Information Processing Systems 27, pp.
  2321--2329
\newblock
\newblock  (2014)

\bibitem{Shrivastava2015}
Shrivastava, A., Li, P.: {Improved asymmetric locality sensitive hashing (ALSH)
  for maximum inner product search (MIPS)}.
\newblock In: Proceedings of the Conference on Uncertainty in Artificial
  Intelligence
\newblock
\newblock  (2015)

\bibitem{Stock2014}
Stock, M., Fober, T., H{\"{u}}llermeier, E., Glinca, S., Klebe, G., Pahikkala,
  T., Airola, A., {De Baets}, B., Waegeman, W.: {Identification of functionally
  related enzymes by learning-to-rank methods}.
\newblock IEEE Transactions on Computational Biology and Bioinformatics
  \textbf{11}(6), 1157--1169
\newblock
\newblock  (2014)

\bibitem{Su2009}
Su, X., Khoshgoftaar, T.M.: {A survey of collaborative filtering techniques}.
\newblock Advances in Artificial Intelligence \textbf{2009}, 1--19
\newblock
\newblock  (2009)

\bibitem{Takacs2008}
Tak{\'{a}}cs, G., Pil{\'{a}}szy, I., N{\'{e}}meth, B., Tikk, D.: {Matrix
  factorization and neighbor based algorithms for the netflix prize problem}.
\newblock In: Proceedings of the 2008 ACM conference on Recommender systems,
  pp. 267--274. ACM Press, New York, New York, USA
\newblock
\newblock  (2008)

\bibitem{Tipping1997}
Tipping, M., Bishop, C.: {Probabilistic principal component analysis}.
\newblock Journal of the Royal Statistical Society \textbf{3}, 611--622
\newblock
\newblock  (1997)

\bibitem{Tsoumakas2007}
Tsoumakas, G., Katakis, I.: {Multi-label classification: an overview}.
\newblock International Journal of Data Warehousing {\&} Mining \textbf{3}(3),
  1--13
\newblock
\newblock  (2007)

\bibitem{Vert2007}
Vert, J.P., Qiu, J., Noble, W.S.: {A new pairwise kernel for biological network
  inference with support vector machines}.
\newblock BMC Bioinformatics \textbf{8}(S-10), 1--10
\newblock
\newblock  (2007)

\bibitem{Volinsky2009}
Volinsky, C.: Matrix factorization techniques for recommender systems.
\newblock pp. 30--37
\newblock
\newblock  (2009)

\bibitem{Waegeman2012}
Waegeman, W., Pahikkala, T., Airola, A., Salakoski, T., Stock, M., {De Baets},
  B.: {A kernel-based framework for learning graded relations from data}.
\newblock IEEE Transactions on Fuzzy Systems \textbf{20}(6), 1090--1101
\newblock
\newblock  (2012)

\bibitem{Wang2015}
Wang, C., Liu, J., Luo, F., Deng, Z., Hu, Q.N.: Predicting target-ligand
  interactions using protein ligand-binding site and ligand substructures.
\newblock BMC Systems Biology \textbf{9}(S-1), S2
\newblock
\newblock  (2015)

\bibitem{Weston2006}
Weston, J., Chapelle, O., Elisseeff, A., Sch{\"{o}}lkopf, B., Vapnik, V.:
  {Kernel dependency estimation}.
\newblock In: Advances in Neural Information Processing Systems, vol.~39, pp.
  440--50
\newblock
\newblock  (2006)

\bibitem{Yamanishi2010}
Yamanishi, Y., Kotera, M., Kanehisa, M., Goto, S.: {Drug-target interaction
  prediction from chemical, genomic and pharmacological data in an integrated
  framework.}
\newblock Bioinformatics \textbf{26}(12), i246--54
\newblock
\newblock  (2010)

\bibitem{Zobel2006}
Zobel, J., Moffat, A.: {Inverted files for text search engines}.
\newblock ACM Computing Surveys \textbf{38}(2)
\newblock
\newblock  (2006)

\end{thebibliography}

\end{document}